\documentclass[11pt]{article}

\usepackage{amsthm,amsfonts,amsmath,algonjk}
\usepackage{graphicx,url,xspace}
\usepackage{picins,pifont,boxedminipage}
\usepackage{tikz}
\usepackage{hyperref}

\def\Comment#1{\textsl{$\langle\!\langle$#1\/$\rangle\!\rangle$}}

\advance \topmargin by -1.0in
\advance \oddsidemargin by -0.8in
\newcommand{\myparskip}{3pt}
\parskip \myparskip

\newtheorem{lemma}{Lemma}[section]
\newtheorem{theorem}[lemma]{Theorem}

\newtheorem{definition}[lemma]{Definition}

\newtheorem{prop}[lemma]{Proposition}
\newtheorem{claim}[lemma]{Claim}

\newcommand{\yes}{{\sc Yes-Instance}\xspace}
\newcommand{\no}{{\sc No-Instance}\xspace}

\newenvironment{note}[1]{\medskip\noindent \textbf{#1:}}%
        {\medskip}

\newcommand{\opt}{\textsc{OPT}}
\newcommand{\etal}{{\em et al.}\ }
\newcommand{\ignore}[1]{}
\newcommand{\mc}[2]{\multicolumn{#1}{c}{#2}}

\def\floor#1{\lfloor {#1} \rfloor}
\def\ceil#1{\lceil {#1} \rceil}
\def\script#1{\mathcal{#1}}

\def\stdLP{\eqref{eq:std}\xspace}
\def\KCLP{\eqref{eq:kc}\xspace}
\def\kway{$k$-Way--$\script{R}$-Connected Subgraph\xspace}
\def\P{\script{P}}
\def\C{\script{C}}

\def\leader{\texttt{Leader}}
\def\hleader{h\texttt{-Leader}}
\def\class{{\tt class}} 
\def\capR{Cap-$R$-Connected Subgraph\xspace}

\renewenvironment{proof}{\vspace{-0.1in}\noindent{\bf Proof:}}%
        {\hspace*{\fill}$\Box$\par}
\newenvironment{proofof}[1]{\smallskip\noindent{\bf Proof of #1:}}%
        {\hspace*{\fill}$\Box$\par}
        {\hspace*{\fill}$\Box$\par}

\title{Approximability of Capacitated Network Design}
\author{
Deeparnab Chakrabarty\thanks{Department of Computer and Information Science, University of Pennsylvania,
Philadelphia PA. Email: {\tt deepc@seas.upenn.edu}.}
\and
Chandra Chekuri\thanks{Dept. of Computer Science, University of Illinois,
Urbana, IL 61801. Partially supported by NSF grants CCF-0728782 and
CNS-0721899. {\tt chekuri@cs.illinois.edu}}
\and
Sanjeev Khanna \thanks{Dept. of Computer \& Information Science, University of Pennsylvania,
Philadelphia, PA 19104. Supported in
part by NSF Awards CCF-0635084 and IIS-0904314. {\tt sanjeev@cis.upenn.edu}}
\and
Nitish Korula\thanks{Dept. of Computer Science, University of Illinois, Urbana,
    IL 61801. Partially supported by NSF grant CCF 0728782 and a
    University of Illinois Dissertation Completion Fellowship. {\tt
      nkorula2@illinois.edu}}
}
\date{}
\begin{document}
\maketitle
\begin{abstract}
  In the {\em capacitated} survivable network design problem
  (Cap-SNDP), we are given an undirected multi-graph where each edge
  has a capacity and a cost.  The goal is to find a minimum cost
  subset of edges that satisfies a given set of pairwise minimum-cut
  requirements.  Unlike its classical special case of SNDP when all
  capacities are unit, the approximability of Cap-SNDP is not well
  understood; even in very restricted settings no known algorithm
  achieves a $o(m)$ approximation, where $m$ is the number of edges in
  the graph. In this paper, we obtain several new results and insights
  into the approximability of Cap-SNDP.

  We give an $O(\log n)$ approximation for a special case of Cap-SNDP
  where the global minimum cut is required to be at least $R$. (Note
  that this problem generalizes the min-cost $\lambda$-edge-connected
  subgraph problem, which is the special case of our problem when all
  capacities are unit.)  Our result is based on a rounding of a
  natural cut-based LP relaxation strengthened with knapsack-cover
  (KC) inequalities. Our technique extends to give a similar
  approximation for a new network design problem that captures global
  minimum cut as a special case.  We then show that as we move away
  from global connectivity, even for the single pair case (that is,
  when only one pair $(s,t)$ has positive connectivity requirement),
  this strengthened LP has $\Omega(n)$ integrality gap.  Furthermore,
  in directed graphs, we show that single pair Cap-SNDP is
  $2^{\log^{1-\delta} n}$-hard to approximate for any fixed
  constant $\delta>0$.

  We also consider a variant of the Cap-SNDP in which multiple copies
  of an edge can be bought: we give an $O(\log k)$ approximation for
  this case, where $k$ is the number of vertex pairs with non-zero
  connectivity requirement. This improves upon the previously known
  $O(\min\{k,\log R_{\max}\})$-approximation for this problem when the
  largest minimum-cut requirement, namely $R_{\max}$, is large. On the
  other hand, we observe that the multiple copy version of Cap-SNDP is
  $\Omega(\log \log n)$-hard to approximate even for the single-source
  version of the problem.
\end{abstract}
\thispagestyle{empty}
\newpage
\setcounter{page}{1}
\section{Introduction}
In this paper we consider the {\em capacitated} survivable network
design problem (Cap-SNDP). The input consists of an undirected
$n$-vertex multi-graph $G(V,E)$ and an integer requirement $R_{ij}$
for each unordered pair of nodes $(i,j)$. Each edge $e$ of $G$ has a
cost $c(e)$ and an integer capacity $u(e)$. The goal is to find a
minimum-cost subgraph $H$ of $G$ such that for each pair of nodes
$i,j$ the capacity of the minimum-cut between $i$ and $j$ in $H$ is at
least $R_{ij}$. This generalizes the well-known survivable network
design problem (SNDP) problem in which all edge capacities are
$1$. SNDP already captures as special cases a variety of fundamental
connectivity problems in combinatorial optimization such as the
min-cost spanning tree, min-cost Steiner tree and forest, as well as
min-cost $\lambda$-edge-connected subgraph; each of these problems has been
extensively studied on its own and several of these special cases are
NP-hard and APX-hard to approximate. Jain, in an influential paper
\cite{Jain}, obtained a $2$-approximation for SNDP via the standard
cut-based LP relaxation using the iterated rounding technique.

Although the above mentioned $2$-approximation for SNDP has been known
since 1998, the approximability of Cap-SNDP has essentially been wide
open even in very restricted special cases.  Similar to SNDP, Cap-SNDP
is motivated by both practial and theoretical considerations. These
problems find applications in the design of resilient networks such as
in telecommunication infrastructure. In such networks it is often
quite common to have equipment with different discrete capacities;
this leads naturally to design problems such as Cap-SNDP.  At the
outset, we mention that a different and somewhat related problem is
also referred to by the same name, especially in the operations
research literature. In this version the subgraph $H$ has to support
{\em simultaneously} a flow of $R_{ij}$ between each pair of nodes
$(i,j)$; this is more closely related to multicommodity flows and
buy-at-bulk network design. Our version is more related to
connectivity problems such as SNDP.

As far as we are aware, the version of Cap-SNDP that we study was
introduced (in the approximation algorithms literature) by Goemans
\etal \cite{GG+} in conjunction with their work on SNDP. They made
several observations on Cap-SNDP: (i) Cap-SNDP reduces to SNDP if all
capacities are the same, (ii) there is an $O(\min(m, R_{\max}))$
approximation where $m$ is the number of edges in $G$ and $R_{\max} =
\max_{ij} R_{ij}$ is the maximum requirement, and (iii) if {\em
  multiple} copies of an edge are allowed then there is an $O(\log
R_{\max})$-approximation.  We note that in the capacitated case
$R_{\max}$ can be exponentially large in $n$, the number of nodes of
the graph. Carr \etal \cite{CFLP} observed that the natural cut-based
LP relaxation has an unbounded integrality gap even for the graph
consisting of only two nodes $s,t$ connected by parallel edges with
different capacities. Motivated by this observation and the goal of
obtaining improved approximation ratios for Cap-SNDP, \cite{CFLP}
strengthened the basic cut-based LP by using {\em knapsack-cover}
inequalities. (Several subsequent papers in approximation algorithms
have fruitfully used these inequalities.) Using these inequalities,
\cite{CFLP} obtained a $\beta(G)+1$ approximation for Cap-SNDP where
$\beta(G)$ is the maximum cardinality of a \emph{bond} in the
underlying simple graph: a bond is a minimal set of edges that
separates some pair of vertices with positive demand. Although
$\beta(G)$ could be $\Theta(n^2)$ in general, for certain topologies
--- for instance, if the underlying graph is a line or a cycle ---
this gives constant factor approximations.

The above results naturally lead to several questions.  What is the
approximability of Cap-SNDP? Should we expect a poly-logarithmic
approximation or even a constant factor approximation?  If not, what
are interesting and useful special cases to consider?  And do the
knapsack cover inequalities help in the general case?  What is the
approximability of Cap-SNDP if one allows multiple copies? Does this
relaxed version of the problem allow a constant factor
approximation?

In this paper we obtain several new positive and negative results for 
Cap-SNDP that provide new insights into the questions above.

\subsection{Our Results}
We first discuss results for Cap-SNDP where multiple copies are not
allowed.  We initiate our study by considering the {\em global
  connectivity} version of Cap-SNDP where we want a min-cost subgraph
with global min-cut at least $R$; in other words, there is a
``uniform'' requirement $R_{ij} = R$ for {\em all} pairs $(i,j)$. We
refer to this as the {\em \capR} problem; the special case when all
capacities are unit corresponds to the classical minimum cost
$\lambda$-edge-connected (spanning) subgraph problem, which 
is known to be APX-hard \cite{Fer}.  We show the following positive
result for arbitrary capacities.

\begin{theorem}\label{thm:uniform}
  There is a randomized $O(\log n)$-approximation algorithm for the
  \capR problem.  Moreover, for any $\gamma \ge 1$, there is a
  randomized $O(\gamma \log n)$-approximation algorithm with running
  time $n^{O(\gamma)}$ for ``nearly uniform'' Cap-SNDP when all
  pairwise requirements are in $[R, \gamma R]$.
\end{theorem}

To prove Theorem \ref{thm:uniform}, we begin with a natural LP
relaxation for the problem. Almost all positive results previously
obtained for the unit capacity case are based on this relaxation. As
remarked already, this LP has an unbounded integrality gap even for a
graph with two nodes (and hence for \capR). We strengthen the
relaxation by adding the valid knapsack cover inequalities.  Although
we do not know of a polynomial time algorithm to separate over these
inequalities, following \cite{CFLP}, we find a violated inequality
{\em only if} the current fractional solution does not satisfy certain
useful properties. Our main technical tool both for finding a violated
inequality and subsequently rounding the fractional solution is
Karger's theorem on the number of small cuts in undirected graphs
\cite{Karger}.

We believe the approach outlined above may be useful in other network
design applications. As a concrete illustration, we use it to solve an
interesting and natural generalization of \capR, namely, the {\em
  \kway} problem.  The input consists of $(k-1)$ integer requirements
$R_1, \ldots R_{k-1}$, such that $R_1 \le R_2 \le \ldots \le
R_{k-1}$. The goal is to find a minimum-cost subgraph $H$ of $G$ such
that for each $1 \le i \le k-1$, the capacity of any $(i+1)$-way cut
of $G$ is at least $R_i$.\footnote{An $i$-way cut $\C$ of a graph
  $G(V,E)$ is a partition of its vertices into $i$ non-empty sets
  $V_1, \ldots, V_i$; we use $\delta(\C)$ to denote the set of edges
  with endpoints in different sets of the partition $\C$.  The
  \emph{capacity} of an $i$-way cut $\C$ is the total capacity of
  edges in $\delta(\C)$.} It is easy to see that \capR is precisely
the \kway, with $k=2$.  Note that the \kway problem is not a special
case of the general Cap-SNDP as the cut requirements for the former
problem are not expressible as pairwise connectivity constraints.
Interestingly, our techniques for \capR can be naturally extended to
handle the multiway cut requirements, yielding the following
generalization of Theorem~\ref{thm:uniform}.

\begin{theorem}\label{thm:kWay}
  There is a randomized $O(k \log n)$-approximation algorithm for the
  \kway problem with $n^{O(k)}$ running time.
\end{theorem}

We remark that even for the unit-capacity case of this problem, it is
not clear how to obtain a better ratio than that guaranteed by the
above theorem. We discuss more in Section~\ref{subsec:kway}.

Once the pairwise connectivity requirements are allowed to vary
arbitrarily, the Cap-SNDP problem seems to become distinctly harder.
Surprisingly, the difficulty of the general case starts to manifest
even for the simplest representative problem in this setting, where
there is only one pair $(s,t)$ with $R_{st} > 0$; we refer to this as
the {\em single pair} problem.  The only known positive result for
this seemingly restricted case is a polynomial-factor approximation
that follows from the results in \cite{GG+,CFLP} for general Cap-SNDP.
We give several negative results to suggest that this special case may
capture the essential difficulty of Cap-SNDP.  In particular, we start
by observing that the LP with knapsack cover inequalities has an
$\Omega(n)$ integrality gap even for the single-pair
problem.\footnote{In \cite{CFLP} it is mentioned that there is a
  series-parallel graph instance of Cap-SNDP such that the LP with
  knapsack-cover inequalities has an integrality gap of at least
  $\floor{\beta(G)/2}+1$. However, no example is given; it is not
  clear if the gap applied to a single pair instance or if $\beta(G)$
  could be as large as $n$ in the construction.} Next we show that
the single pair problem is $\Omega(\log \log n)$-hard to approximate.

\begin{theorem}\label{thm:singlePairHardness}
  The single pair Cap-SNDP problem cannot be approximated to a factor
  better than $\Omega(\log \log n)$ unless $NP \subseteq
  DTIME(n^{{\log \log \log n}})$.
\end{theorem}

The above theorem is a corollary of the results in Chuzhoy \etal's
work on the hardness of related network design problems
\cite{CGNS}. We state it as a theorem to highlight the status of the
problem. (See Appendix~\ref{app:copiesHardness} for a brief proof
sketch.)  We further discuss this connection at the end of this
section. We prove a much stronger negative result for the single pair
problem in {\em directed} graphs. Since in the unit-capacity case,
polynomial-time minimum-cost flow algorithms solve the single-pair
problem exactly even in directed graphs, the hardness result below
shows a stark contrast between the unit-capacity and the non-unit
capacity cases.

\begin{theorem}\label{thm:stHardness}
  In {\em directed} graphs, the single pair Cap-SNDP cannot be approximated to a
  factor better than $2^{\log^{(1-\delta)} n}$ for any $0< \delta <1$,
  unless $NP \subseteq DTIME(n^{{\tt polylog} (n)})$. Moreover, this hardness
  holds for instances in which there are only two distinct edge capacities.
\end{theorem}

\paragraph{Allowing Multiple Copies:} Given the negative results above
for even the special case of the single-pair Cap-SNDP, it is natural
to consider the relaxed version of the problem where multiple copies
of an edge can be chosen. Specifically, for any integer $\alpha \ge 0$,
$\alpha$ copies of $e$ can be bought at a cost of $\alpha\cdot c(e)$ to
obtain a capacity $\alpha\cdot u(e)$.  In some applications, such as in
telecommunication networks, this is a reasonable model. As we
discussed, this model was considered by Goemans \etal \cite{GG+} who
gave an $O(\log R_{\max})$ approximation for Cap-SNDP. This follows
from a simple $O(1)$ approximation for the case when all requirements
are in $\{0,R\}$.  The advantage of allowing multiple copies is that
one can group request pairs into classes and separately solve the
problem for each class while losing only the number of classes in the
approximation ratio.  For instance, one easily obtains a
$2$-approximation for the single pair problem even in directed graphs,
in contrast to the difficulty of the problem when multiple copies are
not allowed. Note that this also implies an easy $2k$ approximation
where $k$ is the number of pairs with $R_{ij} > 0$. We address the
approximability of Cap-SNDP with multiple copies of edges allowed.
When $R_{\max}$ is large, we improve the $\min\{2k, O(\log
R_{\max})\}$-approximation discussed above via the following.

\begin{theorem}\label{thm:multipleCopies}
  In undirected graphs, there is an $O(\log k)$-approximation
  algorithm for Cap-SNDP with multiple copies, where $k$ is the number
  of pairs with $R_{ij} > 0$.
\end{theorem}

Both our algorithm and analysis are inspired by the
$O(\log k)$-competitive online algorithm for the Steiner forest
problem by Berman and Coulston~\cite{BC}, and the subsequent
adaptation of these ideas for the priority Steiner forest problem by
Charikar \etal \cite{CNS}. However, we believe the analysis of our
algorithm is more transparent (although it gets weaker constants) than
the original analysis of \cite{BC}.

We complement our algorithmic result by showing that the multiple copy
version is $\Omega(\log \log n)$-hard to approximate.  This hardness
holds even for the {\em single-source} Cap-SNDP where we are given a
source node $s\in V$, and a set of terminals $T\subseteq V$, such that
$R_{ij} > 0$ iff $i=s$ and $j\in T$.  Observe that single-source
Cap-SNDP is a simultaneous generalization of the classical Steiner
tree problem ($R_{ij} \in \{0,1\}$) as well as both \capR and
single-pair Cap-SNDP.

\begin{theorem}\label{thm:copiesHardness}
  In undirected graphs, single source Cap-SNDP with multiple copies
  cannot be approximated to a factor better than $\Omega(\log \log n)$
  unless $NP \subseteq DTIME(n^{{\log \log \log n}})$.
\end{theorem}

The above theorem, like Theorem~\ref{thm:singlePairHardness}, also
follows easily from the results of \cite{CGNS}.  For completeness, we
provide a proof of Theorem~\ref{thm:copiesHardness} in
Appendix~\ref{app:copiesHardness}. We note that the hardness reduction 
above creates instances with super-polynomially
large capacities. For such instances, our $O(\log k)$-approximation
strongly improves on the previously known approximation guarantees.

\medskip
\noindent {\bf Related Work:} Network design has a large literature in
a variety of areas including computer science and operations
research. Practical and theoretical considerations have resulted in
numerous models and results. Due to space considerations it is
infeasible even to give a good overview of closely related work. We
briefly mention some work that allows the reader to compare the model
we consider here to related models. As we mentioned earlier, our
version of Cap-SNDP is a direct generalization of SNDP and hence is
concerned with (capacitated) connectivity between request node pairs.
We refer the reader to the survey \cite{KortsarzN} and some recent and
previous papers
\cite{GG+,Jain,FleischerJW,ChuzhoyK08,CKsndp,NutovBifamiliesFOCS} for
pointers to literature on network design for connectivity.  A
different model arises if one wishes to find a min-cost subgraph that
supports multicommodity flow for the request pairs; in this model each
node pair $(i,j)$ needs to routes a flow of $R_{ij}$ in the chosen
graph and these flows simultaneously share the capacity of the graph.
We observe that if multiple copies of an edge are allowed then this
problem is essentially equivalent to the non-uniform buy-at-bulk
network design problem. Buy-at-bulk problems have received substantial
attention; we refer the reader to \cite{CHKS} for several pointers to
this work. If multiple copies are not allowed, the approximability of
this flow version is not well-understood; for example if the flow for
each pair is only allowed to be routed on a single path, then even
checking feasibility of a given subgraph is NP-Hard since the problem
captures the well-known edge-disjoint paths and unsplittable flow
problems. Andrews and Zhang \cite{AZ} have recently considered special
cases of this problem with uniform capacities while allowing some
congestion (that is, a few copies) on the chosen edges.

The \kway problem that we consider does not appear to have 
been considered previously even in the unit-capacity case.

\section{The \capR problem}
\label{sec:uniformReq}

In this section, we prove Theorem \ref{thm:uniform}, giving an $O(\log
n)$-approximation for the \capR problem. Here, we assume each $R_{ij}
= R$; the extension to the case when requirements are ``nearly
uniform'' is deferred to Appendix~\ref{app:nearlyUniform}.  We start
by writing a natural linear program relaxation for the problem; the
integrality gap of this LP can be arbitrarily large. To deal with
this, we introduce additional valid inequalities, called the
\emph{knapsack cover} inequalities, that must be satisfied by any
integral solution. We show how to round this strengthened LP,
obtaining an $O(\log n)$-approximation.

\subsection{The Standard LP Relaxation and Knapsack-Cover  Inequalities}

We assume without any loss of generality that the capacity of
any edge is at most $R$.
For each subset $S \subseteq 2^V$, we use $\delta(S)$ to denote the
set of edges with exactly one endpoint in $S$. For a set of edges $A$,
we use $u(A)$ to denote $\sum_{e \in A} u(e)$. We say that a set of
edges $A$ satisfies (the cut induced by) $S$ if $u(A \cap \delta(S))
\ge R$. Note that we wish to find the cheapest set of edges which
satisfies every subset $\emptyset \neq S\subset V$. The following is
the LP relaxation of the standard integer program capturing the
problem.  \vspace{-5mm}
\begin{align}
\min ~~~ \sum_{e\in E} c(e)x_e \tag{Std LP} \label{eq:std} \\
\vspace{-2mm}
\forall S\subseteq V, ~~~~~~~~~~~~ \sum_{e\in \delta(S)} u(e)x_e \ge R \notag \\
\vspace{-3mm}
\forall e\in E,~~~~~~~~~~~~~~~~~~~~~ 0 \le x_e \le 1 \notag
\end{align}

The following example shows that \stdLP can have integrality gap as bad as $R$.

\begin{note}{Example 1}
  Consider a graph $G$ on three vertices $p,q,r$. Edge $pq$ has cost
  $0$ and capacity $R$; edge $qr$ has cost $0$ and capacity $R-1$; and
  edge $pr$ has cost $C$ and capacity $R$. To achieve a global min-cut
  of size at least $R$, any integral solution must include edge $pr$,
  and hence must have cost $C$. In contrast, in \stdLP one can set
  $x_{pr} = 1/R$, and obtain a total cost of $C/R$.
\end{note}

In the previous example, any integral solution in which the mincut
separating $r$ from $\{p,q\}$ has size at least $R$ must include edge
$pr$, even if $qr$ is selected. The following valid inequalities are
introduced precisely to enforce this condition. More generally, let
$S$ be a set of vertices, and $A$ be an arbitrary set of edges. Define
$R(S,A) = \max\{0,R - u(A \cap \delta(S))\}$ be the \emph{residual}
requirement of $S$ that must be satisfied by edges in
$\delta(S)\setminus A$. That is, any feasible solution has $\sum_{e
  \in \delta(S)\setminus A} u(e) x_e \ge R(S,A)$. However, any
integral solution also satisfies the following stronger requirement
$$\sum_{e \in \delta(S)\setminus A} \min\{R(S,A), u(e)\} x_e \ge R(S,A)$$
and thus these inequalities can be added to the LP to strengthen it.
These additional inequalities are referred to as \emph{Knapsack-Cover}
inequalities, or simply KC inequalities, and were first used
by~\cite{CFLP} in design of approximation algorithms for Cap-SNDP.


Below, we write a LP relaxation, \KCLP, strengthened with the knapsack
cover inequalities. Note that the original constraints correspond to
KC inequalities with $A = \emptyset$; we simply write them explicitly
for clarity.

\vspace{-5mm}
\begin{align}
\min ~~~ \sum_{e\in E} c(e)x_e \tag{KC LP} \label{eq:kc} \\
\vspace{-2mm}
\forall S\subseteq V, ~~~~~~~~~~~~ \sum_{e\in \delta(S)} u(e)x_e \ge R \tag{Original Constraints} \\
\vspace{-3mm}
\forall A\subseteq E, \forall S\subseteq V, ~~~~~~~~~  \sum_{e\in \delta(S)\setminus A} \min(u(e),R(S,A))x_e \ge R(S,A)
\tag{KC-inequalities} \\
\vspace{-3mm}
\forall e\in E,~~~~~~~~~~~~~~~~~~~~~ 0 \le x_e \le 1 \notag
\end{align}


\noindent
The Linear Program \KCLP, like the original \stdLP, has exponential
size. However, unlike the \stdLP, we do not know of the existence of
an efficient separation oracle for this. Nevertheless, as we show
below, we do not need to solve \KCLP; it suffices to get to what we
call a {\em good} fractional solution. 


\begin{definition}
  Given a fractional solution $x$, we say an edge $e$ is
  {\em nearly integral} if $x_e \ge \frac{1}{40 \log n}$, and
  we say $e$ is {\em highly fractional} otherwise.
\end{definition}

\begin{definition}
  For any $\alpha \ge 1$, a cut in a graph $G$ with capacities on
  edges, is an \emph{$\alpha$-mincut} if its capacity is within a
  factor $\alpha$ of the minimum cut of $G$.
\end{definition}

\begin{theorem}\label{thm:cutCounting}
  {\rm [Theorems 4.7.6 and 4.7.7 of \cite{Karger}]}
  The number of $\alpha$-mincuts in an $n$-vertex graph is at most
  $n^{2 \alpha}$. Moreover, the set of all $\alpha$-mincuts can be found in
  $O(n^{2 \alpha} \log^2 n)$ time with high probability.
\end{theorem}
\def\u{\hat{u}} \def\S{\script{S}} Given a fractional solution $x$ to
the edges, we let $A_x$ denote the set of nearly integral edges, that
is, $A_x := \{e\in E: x_e \ge \frac{1}{40\log n}\}$.  Define
$\hat{u}(e) = u(e)x_e$ to be the fractional capacity on the edges.
Let $\S := \{S\subseteq V: \u(\delta(S)) \le 2R\}$.  A solution $x$
is called {\em good} if it satisfies the following three
conditions:

\begin{enumerate}
\item[(a)] The global mincut in $G$ with capacity $\u$ is at least
  $R$, i.e. $x$ satisfies the original constraints.
\item[(b)] The KC inequalities are satisfied for the set $A_x$ and the
  sets in $\S$. Note that if (a) is satisfied, then by Theorem
  \ref{thm:cutCounting}, $|\S| \le n^4$.
\item[(c)] $\sum_{e\in E} c(e)x_e$ is at most the value of the optimum
  solution to \KCLP.
\end{enumerate}

Note that a good solution need not be {\em feasible} for \KCLP as it is required
to satisfy only a subset of KC-inequalities. We use the ellipsoid method to
get such a solution. Such a method was also used in \cite{CFLP}.

\begin{lemma}\label{lem:ell}
There is a randomized algorithm that computes a good fractional solution with high probability.
\end{lemma}
\begin{proof}
  We start by guessing the optimum value $M$ of \KCLP and add the
  constraint $\sum_{e\in E}c(e)x_e \le M$ to the constraints of
  \KCLP. If the guessed value is too small, a good solution may not
  exist; however, a simple binary search suffices to identify the
  smallest feasible value of $M$.  With this constraint in place, we
  will use the ellipsoid method to compute a solution that satisfies
  (a), (b), and (c) with high probability.  Since we do not know of a
  polynomial-time separation oracle for KC inequalities, we will
  simulate a separation oracle that verifies condition (b), a subset
  of KC inequalities, in polynomial time. Specifically, we give a
  randomized polynomial time algorithm such that given a solution $x$
  that violates condition (b), the algorithm detects the violation
  with high probability and outputs a violated KC inequality. We now
  describe the entire process.

  Given a solution $x$ we first check if condition (a) is
  satisfied. This can be done in polynomial time by $O(n)$ max-flow
  computations. If (a) is not satisfied, we have found a violated
  constraint. Once we have a solution that satisfies (a), we know that
  $|\S| \le n^4$. By Theorem~\ref{thm:cutCounting}, the set $\S$ can
  be computed in polynomial-time with high probability.  Thus we can
  check condition (b) in polynomial-time, and with high-probability
  find a violating constraint for (b) if one exists. Once we have a
  solution that satisfies both (a) and (b), we check if $\sum_{e\in
    E}c(e)x_e \le M$. If not, we have once again found a violated
  constraint for input to the ellipsoid algorithm.  Thus in
  polynomially many rounds, where each round runs in polynomial-time,
  the ellipsoid algorithm combined with the simulated separation
  oracle, either returns a solution $x$ that satisfies (a), (b), and
  $\sum_{e\in E}c(e)x_e \le M$, with high probability, or proves that
  the system is infeasible. Using binary search, we find the smallest
  $M$ for which a solution $x$ is returned satisfying conditions (a),
  (b) and $\sum_{e\in E}c(e)x_e \le M$. Since $M$ is less than the
  optimum value of \KCLP, we get that the returned $x$ is a good
  fractional solution with high probability.
\end{proof}



\subsection{The Rounding and Analysis}

Given a good fractional solution $x$, we now round it to get a $O(\log
n)$ approximation to the \capR problem.  A useful tool for our
analysis is the following Chernoff bound (see \cite{MR}, for
instance):

\begin{lemma}\label{lem:Chernoff}
  Let $X_1, X_2, \ldots X_k$ be a collection of independent random
  variables in $[0,1]$, let $X = \sum_{i=1}^k X_i$, and let $\mu =
  \mathbb{E}[X]$. The probability that $X \le (1-\delta)\mu$ is at
  most $e^{-\mu \delta^2/2}$.
\end{lemma}

We start by selecting $A_x$, the set of all nearly integral edges.
Henceforth, we lose the subscript and denote the set as simply $A$.
Let $F = E \setminus A$ denote the set of all highly fractional edges;
for each edge $e \in F$, select it with probability $(40 \log n \cdot
x_e)$. Let $F^* \subseteq F$ denote the set of selected highly
fractional edges. The algorithm returns the set of edges $E_A := A\cup
F^*$.

It is easy to see that the expected cost of this solution $E_A$ is
$O(\log n) \sum_{e\in E}c(e)x_e$, and hence by condition (c) above,
within $O(\log n)$ times that of the optimal integral solution. Thus,
to prove Theorem~\ref{thm:uniform}, it suffices to prove that with high
probability, $E_A$ satisfies every cut in the graph $G$; we devote the
rest of the section to this proof.  We do this by separately
considering cuts of different capacities, where the capacities are
w.r.t $\u$ (recall that $\u(e) = u(e)x_e$). Let ${\mathcal L}$ be the set of
cuts of capacity at least $2R$, that is, ${\mathcal L} := \{S\subseteq
V: \u(\delta(S)) > 2R\}$.

%
\def\L{{\mathcal L}}
\begin{lemma}\label{lem:largeCuts}
  ${\bf Pr} [~\forall S\in {\mathcal L}: ~u(E_A \cap \delta(S))\ge R]
  ~\ge ~ 1-\frac{1}{2n^{10}}$.
\end{lemma}
\begin{proof}
  We partition $\L$ into sets $\L_2,\L_3,\cdots $ where
  $\L_j := \{S\subseteq V: jR < \u(\delta(S)) \le (j+1)R\} .$
  Note that
  Theorem~\ref{thm:cutCounting} implies $|\L_j| \le n^{2(j+1)}$ by
  condition (a) above.
  Fix $j$, and consider an arbitrary cut $S\in
  \L_j$. 
  If $u(A \cap \delta(S)) \ge R$, then $S$ is clearly satisfied by
  $E_A$. Otherwise, since the total $\u$-capacity of $S$ is at least
  $jR$, we have $\u(F\cap \delta(S)) \ge \u(\delta(S)) - u(A\cap
  \delta(S)) \ge (j-1)R$. Thus

  $$\sum_{e\in F\cap \delta(S)} \frac{u(e)}{R}x_e \ge (j-1)$$
  Recall that an edge $e\in F$ is selected in $F^*$ with probability
  $(40\log n\cdot x_e)$. Thus, for the cut $S$, the expected value of
  $\sum_{e\in F^*\cap\delta(S)} \frac{u(e)}{R} \ge 40(j-1)\log
  n$. Since $u(e)/R \le 1$, we can apply Lemma \ref{lem:Chernoff} to
  get that the probability that $S$ is not satisfied is at most
  $e^{-16 \log n (j-1)} = 1/n^{16(j-1)}$.  Applying the union bound, the
  probability that there exists a cut in $\L_j$ not satisfied by $E_A$
  is at most $n^{2(j+1)}/n^{16(j-1)} = n^{18-14j}$. Thus probability that
  some cut in $\L$ is not satisfied is bounded by $\sum_{j \ge 2} n^{18-14j} \le 2n^{-10}$ if $n\ge 2$.
  Hence with probability at least $1 - 1/2n^{10}$,
  $A \cup F^*$ satisfies all cuts in $\L$.
\end{proof}

One might naturally attempt the same approach for the cuts in $\S$
(recall that ${\mathcal S} = \{S\subseteq V:
\u(\delta(S)) \le 2R\}$) modified as follows. Consider any cut $S$, which is partly satisfied
by the nearly integral edges $A$. The fractional edges contribute to
the residual requirement of $S$, and since $x_e$ is scaled up for
fractional edges by a factor of $40 \log n$, one might expect that
$F^*$ satisfies the residual requirement, with the $\log n$ factor
providing a high-probability guarantee. This intuition is correct, but
the KC inequalities are crucial. Consider Example 1; edge $pr$ is
unlikely to be selected, even after scaling. In the statement of
Lemma~\ref{lem:Chernoff}, it is important that each random variable
takes values in $[0,1]$; thus, to use this lemma, we need the expected
capacity from fractional edges to be large compared to the maximum
capacity of an individual edge. But the KC inequalities,
in which edge capacities are ``reduced'', enforce precisely this
condition. Thus we get the following lemma using a similar analysis as above.

\begin{lemma}\label{lem:smallCuts}
  ${\bf Pr} [~\forall S\in {\mathcal S}: ~u(\delta(E_A \cup
  \delta(S)))\ge R] ~\ge ~ 1-\frac{1}{n^{12}}$.
\end{lemma}
\noindent

The $O(\log n)$-approximation guarantee for the \capR problem stated
in Theorem~\ref{thm:uniform} follows from the previous two lemmas.

\subsection{The \kway Problem}
\label{subsec:kway}

The \kway problem that we define is a natural generalization of the
well-studied min-cost $\lambda$-edge-connected subgraph problem. The latter
problem is motivated by applications to fault-tolerant network design
where any $\lambda-1$ edge failures should not disconnect the graph.
However, there may be situations in which global $\lambda$-connectivity may
be too expensive or infeasible. For example the underlying graph $G$
may have a single cut-edge but we still wish a subgraph that is as
close to $2$-edge-connected as possible. We could model the
requirement by \kway (in the unit-capacity case) by setting $R_1 = 1$
and $R_2 = 3$; that is, at least $3$ edges have to be removed to 
partition the graph into $3$ disconnected pieces.

We briefly sketch the proof of Theorem~\ref{thm:kWay}.
We work with a generalization of \KCLP to $i$-way cuts, with an
original constraint for each $i+1$-way cut, $1 \le i \le k-1$, and
with KC inequalities added. The algorithm is to select all nearly
integral edges $e$ (those with $x_e \ge \frac{1}{40 k \log n}$), and
select each of the remaining (highly fractional) edges $e$ with
probability $40 k \log n \cdot x_e$.  The analysis is very similar to
that of Theorem~\ref{thm:uniform} and hence moved to the
appendix; we use the following lemma on
counting $k$-way cuts in place of Theorem~\ref{thm:cutCounting}. 


\begin{lemma}[Lemma 11.2.1 of \cite{Karger}]\label{lem:kWayCutCounting}
  In an $n$-vertex graph, the number of $k$-way cuts with capacity at
  most $\alpha$ times that of a minimum $k$-way cut is at most $n^{2
    \alpha (k-1)}$. 
\end{lemma}

It would be interesting to explore algorithms and techniques for other
more general variants of the \kway problem that we consider here.

\section{Single-Pair Cap-SNDP}
\label{sec:single-pair}

In this section we show that the integrality gap with KC inequalities
is $\Omega(n)$ even for single-pair Cap-SNDP in undirected
graphs. Moreover, when the underlying graph is directed, we show that
the single-pair problem is hard to approximate to within a factor of
$2^{\log^{(1-\delta)} n}$ for any $\delta > 0$.

\subsection{Integrality Gap with KC Inequalities}
\label{sec:kcBad}

We show that for any positive integer $R$, there exists a single-pair
Cap-SNDP instance $G$ with $(R+2)$ vertices such that the integrality
gap of the natural LP relaxation strengthened with KC inequalities is
$\Omega(R)$. The instance $G$ consists of a source vertex $s$, a sink
vertex $t$, and $R$ other vertices $v_1, v_2, \ldots, v_R$.

\parpic[r]{
  \begin{boxedminipage}{0.27\linewidth}
    \begin{tikzpicture}[xscale=0.92,yscale=1.1]
      \tikzstyle{dot}=[circle,inner sep=1pt,fill=black];
      \tikzstyle{vtx}=[circle,draw,minimum size=7mm,fill=white];
      \tikzstyle{every node}=[font=\footnotesize];
      \node (s) at (-2,0) [vtx] {$s$};
      \node (t) at (2,0) [vtx] {$t$};
      \foreach \y in {-2, 1,2} {
        \draw (s) -- (0,\y) -- (t);
      }
      \foreach \y in {0,-0.5,-1} {
        \node at (0,\y) [dot] {};
      }
      \node at (0,2) [vtx] {$v_1$};
      \node at (0,1) [vtx] {$v_2$};
      \node at (0,-2) [vtx] {$v_{R}$};
      \node at (-1.35,1.25) {(2,1)}; \node at (1.5,1.25) {(R,R)};
      \node at (-1.35,-1.25) {(2,1)}; \node at (1.5,-1.25) {(R,R)};
      \node at (-1,0.1) {(2,1)}; \node at (1,0.1) {(R,R)};
    \end{tikzpicture}
  \end{boxedminipage}
}

\noindent There is an edge of capacity $2$ and cost $1$ (call these
\emph{small} edges) between $s$ and each $v_i$, and an edge of
capacity $R$ and cost $R$ between each $v_i$ and $t$ (\emph{large}
edges). We have $R_{st} = R$. Clearly, an optimal integral solution
must select at least $R/2$ of the large edges (in addition to small
edges), and hence has cost greater than $R^2/2$.  The instance is
depicted in the accompanying figure: Label $(u,c)$ on an edge denotes
capacity $u$ and cost $c$.


We now describe a feasible LP solution: set $x_e = 1$ on each small
edge $e$, and $x_{e'} = 2/R$ on each large edge $e'$. The cost of this
solution is $R$ from the small edges, and $2R$ from the large edges,
for a total of $3R$. This is a factor of $R/6$ smaller
than the optimal integral solution, proving the desired integrality
gap.

\medskip
It remains only to verify that this is indeed a feasible solution to
\KCLP. Consider the constraint corresponding to sets $S,A$. As edges
in $A \setminus \delta(S)$ play no role, we may assume $A \subseteq
\delta(S)$. If $A$ includes a large edge, or at least $R/2$ small
edges, the residual requirement $R(S,A)$ that must be satisfied by the
remaining edges of $\delta(S)$ is $0$, and so the constraint is
trivially satisfied. Let $A$ consist of $a < R/2$ small edges; the
residual requirement is thus $R-2a$. Let $\delta(S)$ contain $i$ large
edges and thus $R - i$ small edges. Now, the contribution to the left
side of the constraint from small edges in $\delta(S) \setminus A$ is
$2(R - i - a) = (R-2a) + (R - 2i)$. Therefore, the residual
requirement is satisfied by small edges alone unless $i > R/2$.  But
the contribution of large edges is $i \cdot \frac{2}{R} \cdot (R-2a)$
which is greater than $R-2a$ whenever $i > R/2$. Thus, we satisfy each
of the added KC inequalities.

\subsection{Hardness of Approximation in Directed Graphs}
\label{subsec:stHardness}

We now prove Theorem \ref{thm:stHardness} via a reduction from the label
cover problem~\cite{ABSS}.

\begin{definition} [{\em Label Cover Problem}]
  The input consists of a bipartite graph $G(A\cup B, E)$ such that the degree
  of every vertex in $A$ is $d_A$ and degree of every vertex in $B$ is
  $d_B$,  a set of labels $L_A$ and a set of labels $L_B$, and  a
  relation $\pi_{(a,b)} \subseteq L_A\times L_B$ for each
  edge $(a,b)\in E$.  Given a labeling $\phi: A \cup B \to L_A \cup L_B$,
  an edge $e = (a,b)\in E$ is said to be {\em consistent} iff
  $(\phi(a),\phi(b)) \in \pi_{(a,b)}$. The goal is to find a
  labeling that maximizes the fraction of consistent edges.
\end{definition}

The following hardness result for the label-cover problem is a well-known
consequence of the PCP theorem~\cite{ALMSS} and Raz's Parallel Repetition
theorem~\cite{Raz98}.

\begin{theorem}[{\rm \cite{ALMSS,Raz98}}]
\label{thm:labelcover_hard}
For any $\epsilon > 0$, there does not exist a poly-time algorithm
to decide if a given instance of label cover problem
has a labeling where all edges are consistent {\rm (\yes)}, or
if no labeling can make at least $\frac{1}{\gamma}$ fraction
of edges to be consistent for $\gamma = 2^{\log^{1-\epsilon} n}$ {\rm (\no)},
unless $NP\subseteq DTIME(n^{{\tt polylog}(n)})$.
\end{theorem}

We now give a reduction from label cover
to the single-pair Cap-SNDP in directed graphs. 
In our reduction, the only non-zero capacity values will be $1$, $d_A$, and $d_B$. 
We note that Theorem~\ref{thm:labelcover_hard} holds even when we restrict 
to instances with $d_A = d_B$. Thus our hardness result will hold on single-pair Cap-SNDP
instances where there are only two distinct non-zero capacity values.

Given an instance $I$
of the label cover problem with $m$ edges, we create in polynomial-time a directed
instance $I'$ of single-pair Cap-SNDP such that if $I$ is a \yes
then $I'$ has a solution of cost at most $2m$, and otherwise,
every solution to $I'$ has cost
$\Omega(m\gamma^{\frac{1}{4}})$. This establishes Theorem \ref{thm:stHardness}
when we choose $\epsilon = \delta/2$.

The underlying graph $G'(V',E')$ for the single-pair Cap-SNDP instance
is constructed as follows.  The set $V'$ contains a vertex $v$ for every
$v \in A \cup B$. We slightly abuse notation and refer to
these sets of vertices in $V'$ as $A$ and $B$ as well.
Furthermore, for every vertex $a \in A$, and for every label $\ell \in
L_A$, the set $V'$ contains a vertex $a(\ell)$. Similarly, for every vertex $b\in
B$, and for every label $\ell\in L_B$, the set $V'$ contains a vertex
$b(\ell)$. Finally, $V'$ contains  a source vertex $s$ and a sink vertex $t$.
The set $E'$ contains the following directed edges:

 \vspace{-3mm}
\begin{itemize}
\item For each vertex $a$ in $A$, there is an edge from $s$ to the
  vertex $a$ of cost $0$ and capacity $d_A$. For each vertex
  $b\in B$, there is an edge from $b$ to $t$ of
  cost $0$ and capacity $d_B$.

\item For each vertex $a\in A$, and for all labels $\ell$ in $L_A$, there is
  an edge from $a$ to $a(\ell)$ of cost $d_A$ and
  capacity $d_A$. For each vertex $b \in B$, and for all labels $\ell$ in $L_B$,
  there is an edge from $b(\ell)$ to $b$ of cost $d_B$
  and capacity $d_B$. These two types of edges are the only edges
  with non-zero cost.

\item For every edge $(a,b) \in E$, and for every pair of labels
  $(\ell_a,\ell_b)\in \pi_{(a,b)}$, there is an edge from $a(\ell_a)$
  to $b(\ell_b)$ of cost $0$ and capacity $1$.
\end{itemize}
\vspace{-2mm}
\noindent
This completes the description of the network $G'$. The requirement
$R_{st}$ between $s$ and $t$ is $m$, the number of edges in the label cover instance.
It is easy to verify that the size of the graph $G'$ is at most quadratic
in the size of the label cover
instance, and that $G'$ can be constructed in polynomial-time.
The lemmas below analyze the cost of \yes and \no instances.

\begin{lemma}\label{lemm:yes}
If the label cover instance is a \yes, then $G'$ contains a subgraph
of cost $2m$ which can realize a flow of value $m$ from $s$ to $t$.
\end{lemma}
\begin{proof}Let $\phi$ be any labeling that consistently labels all edges in $G(A \cup B,E)$.
Also, let $E_1 \subseteq E'$ be the set of all edges of cost $0$ in $E'$, and
let $E_2 \subseteq E'$ be  the set of edges $\{(a,a(\phi(a)))~|~
  a\in A\} \cup \{(b(\phi(b)),b): b\in B\}$. We claim that $E_1 \cup E_2$ is a feasible
  solution for the single-pair Cap-SNDP instance.
  Note that the total cost of edges in $E_1 \cup E_2$ is $|A|d_A + |B|d_B = 2m$.
  We now exhibit a flow of value $m$ from $s$ to $t$ in $G''(V',E_1 \cup E_2)$.
   A flow of value $d_A$ is sent along the path $s \to a \to
  a(\phi(a))$ for all $a\in A$.  From $a(\phi(a))$, a unit of flow is sent
  to the $d_A$ vertices of the form $\{b(\phi(b))~|~ b\in B \textrm{
    and } (a,b) \in E\}$; this is feasible because $\phi$  consistently
    labels all edges in $E$.  Thus each vertex of the form $b(\phi(b))$
    where $b \in B$ receives $d_B$
  units of flow, since the degree of $b$ is $d_B$ in $G$.  A flow of
  value $d_B$ is sent to $t$ along the path $b(\phi(b)) \to b \to t$. Thus $s$
  sends out a flow of value $|A|d_A = m$, or equivalently, $t$ receives a flow of value
  $|B|d_B = m$.
\end{proof}

\begin{lemma}\label{lem:no}
If the label cover instance is a \no,
then any subgraph of $G'$ that realizes a flow of $m$ units from $s$ to $t$
has cost $\Omega(m\gamma^{\frac{1}{4}})$.
\end{lemma}
\begin{proof}
 Let $\rho = \gamma^{1/4}/2$, and $M=32/15$.  Assume by way of
  contradiction, that there exists a subgraph $G''(V',E'')$ of $G'$ of
  cost strictly less than $\frac{\rho m}{M}$ that realizes $m$ units
  of flow from $s$ to $t$. We say a vertex $a\in A$ is {\em light} if
  the number of edges of the form $\{(a, a(\ell))~|~\ell\in L_A\}$ in
  $G''$ is less than $\rho$. Similarly, we say a vertex $b\in B$ is
  {\em light} if the number of edges of the from
  $\{(b(\ell),b)~|~\ell\in L_B\}$ in $G''$ is less than $\rho$. All
  other vertices in $A \cup B$ are referred to as {\em heavy}
  vertices. Note that at most $1/M$ fraction of vertices in $A$ could
  be heavy, for otherwise the total cost of the edges in $E''$ would
  exceed $\frac{|A|}{M} \cdot\rho \cdot d_A = \frac{\rho
    m}{M}$. Similarly, at most $1/M$ fraction of vertices in $B$ could
  be heavy.

  Now fix any integral $s$-$t$ flow $f$ of value $m$ in $G''$; an
  integral flow exists since all capacities are integers. We start by
  deleting from $G''$ all heavy vertices.  Since at most $1/M$
  fraction of either $A$ or $B$ are deleted, the total residual flow
  in this network is at least $(1 - \frac{2}{M})m = \frac{m}{16}$
  (recall that $M=32/15$) since at most $d_A$ units of flow can
  transit through a vertex in $A$, and at most $d_B$ units of flow can
  transit through a vertex in $B$.

  Let $F$ be a decomposition of the residual flow into unit flow
  paths.  Note that $|F| = m/16$.  By construction of $G'$, every flow
  path $f \in F$ is of the form $s \rightarrow a \rightarrow \ell_a
  \rightarrow \ell_b \rightarrow b \rightarrow t$ where the pair
  $(\ell_a, \ell_b) \in \pi_{(a,b)}$.  We say that an edge $(a,b) \in
  E$ is a {\em good} edge if there is a flow path $f$ of the above
  form, and we say $f$ is a certificate for edge $(a,b)$ being good.
  Note that every flow path $f$ is a certificate of exactly one edge
  $(a,b)$.  We claim that there are at least $\frac{m}{16\rho^2}$ good
  edges in $G$.  It suffices to show that for any edge $(a,b) \in E$,
  at most $\rho^2$ flow paths in $F$ can certify that $(a,b)$ as a
  good edge.  Since $a$ and $b$ are both light vertices, we know that
  $|\{ (a, \ell_a) ~|~\ell_a \in L_A \} \cap E''| \le \rho$ and $|\{
  (\ell_b, b) ~|~\ell_b \in L_B \} \cap E''| \le \rho$. Now using the
  fact that each edge $(\ell_a, \ell_b)$ has unit capacity, it follows
  that at most $\rho^2$ paths in $F$ can certify $(a,b)$ as a good
  edge.  Hence number of good edges in $E$ is at least
  $\frac{m}{16\rho^2}$.

  We now show existence of a labeling $\phi$ that makes at least
  $\frac{1}{\gamma}$ fraction of the edges to be consistent,
  contradicting that we were given a \no of label cover. For a vertex
  $a\in A$, let $\Gamma(a) := \{\ell_a \in L_A~|~ (a,\ell_a) \in E''
  \}$.  Similarly, we define $\Gamma(b)$ for each vertex $b \in B$.
  Consider the following random label assignment: each vertex $a \in
  A$ is assigned uniformly at random a label from $\Gamma(a)$, and
  each vertex in $B$ is assigned uniformly at random a label in
  $\Gamma(b)$. For any good edge $(a,b)$, the probability that the
  random labeling makes it consistent is at least $\frac{1}{\rho^2}$
  since $|\Gamma(a)|$ and $|\Gamma(b)|$ are both less than $\rho$ (as
  $a$ and $b$ are light), and there exists an $\ell_a \in \Gamma_A$
  and $\ell_b \in \Gamma_B$ such that $(\ell_a, \ell_b) \in
  \pi_{(a,b)}$.  Thus, in expectation, at least $\frac{1}{\rho^2}$
  fraction of good edges are made consistent by the random assignment.
  Hence there exists a labeling $\phi$ that $\frac{m}{16\rho^4} =
  \frac{m}{\gamma}$ edges in $G$ consistent.
  \end{proof}
Since the graph $G'$ can be constructed from $G$ in poly-time,
it follows that a poly-time $(\gamma^{1/4}/5)$-approximation algorithm
for single-pair Cap-SNDP would give a poly-time algorithm to
decide whether a given instance of label cover is a \yes or a \no.

\section{Cap-SNDP with Multiple Copies Allowed} \label{sec:multipleCopies}

We now consider the version of Cap-SNDP when multiple copies of any
edge $e$ can be chosen; that is, for any integer $\alpha\ge 0$,
$\alpha$ copies of $e$ can be bought at a cost $\alpha\cdot c(e)$ to
obtain a capacity of $\alpha\cdot u(e)$.  Allowing multiple copies
makes the problem easier, and Goemans \etal \cite{GG+} give a $O(\log
R_{max})$ factor approximation algorithm for the problem.  In this
section, we complement this result with a $O(\log k)$ factor
approximation algorithm, where $k$ is the number of $(i,j)$ pairs with
$R_{ij} > 0$.\footnote{Note that we overload the letter `$k$',
  previously used in the definition of the \kway problem; this should
  cause no ambiguity as we discuss only pairwise connectivity
  requirements in this section.} Our algorithm is inspired by the work
of Berman and Coulston \cite{BC} on online Steiner Forest.  For
notational convenience, we rename the pairs $(s_1,t_1), \cdots,
(s_k,t_k)$, and denote the requirement $R_{s_i,t_i}$ as $R_i$; the
vertices $s_i, t_i$ are referred to as \emph{terminals}. We also
assume that the pairs are so ordered that $R_1 \ge R_2 \ge \cdots \ge
R_k$.

We first give an intuitive overview of the algorithm. The algorithm
considers the pairs in decreasing order of requirements, and maintains
a \emph{forest solution} connecting the pairs that have been already
been processed; that is, if we retain a single copy of each edge in
the partial solution constructed so far, we obtain a forest $F$. For
any edge $e$ on the path in $F$ between $s_j$ and $t_j$, the total
capacity of copies of $e$ will be at least $R_j$. When considering
$s_i, t_i$, we connect them as cheaply as possible, assuming that
edges previously selected for $F$ have $0$ cost. (Note that this can
be done since we are processing the pairs in decreasing order of
requirements and for each edge already present in $F$, the capacity of
its copies is at least $R_i$.) The key step of the algorithm is that
\emph{in addition} to connecting $s_i$ and $t_i$, we also connect the
pair to certain other components of $F$ that are ``nearby''.  The cost
of these additional connections can be bounded by the cost of the
direct connection costs between the pairs. These additional
connections are useful in allowing subsequent pairs of terminals to be
connected cheaply. In particular, they allow us to prove a $O(\log k)$
upper bound on the approximation factor.

We now describe the algorithm in more detail. The algorithm maintains
a forest $F$ of edges that have already been bought; $F$ satisfies the
invariant that, after iteration $i-1$, for each $j \le i-1$, $F$
contains a unique path between $s_j$ and $t_j$. In iteration $i$, we
consider the pair $s_i,t_i$.  We define the cost function $c_i(e)$ as
$c_i(e) := 0$ for edges $e$ already in $F$, and $c_i(e) := c(e) +
\frac{R_i}{u(e)}c(e)$, for edges $e\notin F$. Note that for an edge
$e\notin F$, the cost $c_i(e)$ is sufficient to buy enough copies of
$e$ to achieve a total capacity of $R_i$. Thus it suffices to connect
$s_i$ and $t_i$ and pay cost $c_i(e)$ for each edge; in the Cap-SNDP
solution we would pay at most this cost and get a feasible solution.
However, recall that our algorithm also connects $s_i$ and $t_i$ to 
other ``close'' components; to describe this process, we introduce
some notation:

For any vertices $p$ and $q$, we use $d_i(p,q)$ to denote the distance
between $p$ and $q$ according to the metric given by edge costs
$c_i(e)$. We let $\ell_i := d_i(s_i,t_i)$ be the cost required to
connect $s_i$ and $t_i$, given the current solution $F$. We also
define the \emph{class} of a pair $(s_j, t_j)$, and of a component:
\vspace{-2mm}
\begin{itemize}
  \item For each $j \le i$, we say that pair $(s_j, t_j)$ is in
    \emph{class} $h$ if $2^h \le \ell_j < 2^{h+1}$.\\
    Equivalently, $\class(j) = \floor{\log \ell_j}$.

  \item For each connected component $X$ of $F$, $\class(X) =
    \max_{(s_j, t_j) \in X} \class(j)$. 
\end{itemize}

\vspace{-2mm}
\noindent 
Now, the algorithm connects $s_i$ (respectively $t_i$) to component
$X$ if $d_i(s_i, X)$ (resp.~$d_i(t_i, X)$) $\le 2^{\min\{\class(i),
  \class(X)\}}$. That is, if $X$ is close to the pair $(s_i,t_i)$
compared to the classes they are in, we connect $X$ to the pair. As we
show in the analysis, this extra connection cost can be charged to
some pair $(s_j, t_j)$ in the component $X$. The complete algorithm
description is given below.

\begin{algo}
  \underline{\sc Cap-SNDP-MC}:\\
  $F \leftarrow \emptyset$ \hspace{2.26in}\Comment{$F$ is the forest
    solution returned} \\
  For $i \leftarrow 1$ to $k$ \+ \\
    For each edge $e \in F$, $c_i(e) \leftarrow 0$ \\
    For each edge $e \not \in F$, $c_i(e) \leftarrow c(e) +  (R_i/u(e)) c(e)$ \\
    $\ell_i \leftarrow d_i(s_i, t_i)$ \\
    Add to $F$ a shortest path (of length $\ell_i$) from $s_i$ to
    $t_i$ under distances $c_i(e)$ \\
    $\class(i) \leftarrow \floor{\log \ell_i}$ \\
    For each connected component $X$ of $F$ \+ \\
      If $d_i(s_i, X) \le 2^{\min\{\class(i), \class(X)\}}$ \+ \\
        Add to $F$ a shortest path connecting $s_i$ and $X$ \- \- \\
    For each connected component $X$ of $F$ \+ \\
      If $d_i(t_i, X) \le 2^{\min\{\class(i), \class(X)\}}$ \+ \\
        Add to $F$ a shortest path connecting $t_i$ and $X$ \- \- \\       
    Buy $\ceil{R_i/u_e}$ copies of each edge $e$ added during this iteration. \-
\end{algo}
We prove that this algorithm {\sc Cap-SNDP-MC} gives an $O(\log k)$
approximation. \\

\noindent
The structure of our proof is as follows: Recall that $\ell_i$ was the
direct connection cost between $s_i$ and $t_i$; in addition to paying
$\ell_i$ to connect these vertices, the algorithm also buys additional
edges connecting $s_i$ and $t_i$ to existing components. We first show
(in Lemma~\ref{lem:ellPays}) that the total cost of extra edges bought
can be charged to the direct connection costs; thus, it suffices to
show that $\sum_i \ell_i \le O(\log k) \opt$, where $\opt$ is the cost
of an optimal solution.  To prove this (Lemma~\ref{lem:lowerBound}),
we bucket the pairs $(s_i, t_i)$ into $O(\log k)$ groups based on
$\class(i)$, and show that in each bucket $h$, $\sum_{i: \class(i) =
  h} \ell_i \le O(\opt)$.

\begin{lemma}\label{lem:ellPays}
  The total cost of all edges bought by {\sc Cap-SNDP-MC} is at most
  $9 \sum_{i=1}^k \ell_i$.
\end{lemma}
\begin{proof}
  Let $F_i$ denote the set of edges added to $F$ during iteration
  $i$. First, note the total cost paid for copies of edge $e \in F_i$
  is $\ceil{\frac{R_i}{u(e)}} c(e) < c(e) + \frac{R_i}{u_e}c(e) =
  c_i(e)$. Thus, it suffices to show:
  $$\sum_{i = 1}^k \sum_{e \in F_i}  c_i(e) \le 9 \sum_{i=1}^k \ell_i$$ 

  We prove that the total cost of the \emph{additional} edges bought
  is at most $8\sum_{i=1}^k \ell_i$; this clearly implies the desired
  inequality. It is \emph{not} true that for each $i$, the total cost
  of additional edges bought during iteration $i$ is at most $8 \ell
  i$.  Nonetheless, a careful charging scheme proves the needed bound
  on total cost.  In iteration $i$, suppose we connect the pair
  $(s_i,t_i)$ to the components $X_1,\ldots,X_r$. We charge the cost
  of connecting $(s_i, t_i)$ and component $X_j$ to the connection
  cost $\ell_j$ of a pair $(s_j,t_j)$ in $X_j$. This is possible since
  we know the additional connection cost is at most
  $2^{\class(X_j)}$. Care is required to ensure no pair is
  overcharged. To do so, we introduce some notation.

  At any point during the execution of the algorithm, for any current
  component $X$ of $F$, we let $\leader(X)$ be a pair $(s_i, t_i) \in
  X$ such that $\class(i) = \class(X)$. For integers $h \le
  \class(X)$, $\hleader(X)$ will denote a pair $(s_j, t_j)$ in $X$; we
  explain how this pair is chosen later. (Initially, $\hleader(X)$ is
  undefined for each component $X$.)

  Now, we have to account for additional edges bought during iteration
  $i$; these are edges on a shortest path connecting $s_i$ (or $t_i$)
  to some other component $X$; we assume w.l.o.g.~that the path is
  from $s_i$ to $X$. Consider any such path $P$ connecting $s_i$ to a
  component $X$; we have $\sum_{e \in P} c_i(e) = d_i(s_i, X) \le
  2^{\min\{\class(i), \class(X)\}}$. Let $h = \floor{\log d_i(s_i,
    X)}$: Charge all edges on this path to $\hleader(X)$ if it is
  defined; otherwise, charge all edges on the path to $\leader(X)$. In
  either case, the pair ($s_i, t_i$) becomes the $\hleader$ of the new
  component just formed. Note that a pair $(s_i,t_i)$ could
  simultaneously be the $h_1$-$\leader$, $h_2$-$\leader$, etc. for a
  component $X$ if $(s_i,t_i)$ connected to many components during
  iteration $i$. However, it can never be the $\hleader$ of a
  component for $h>\class(i)$, and once it has been charged as
  $\hleader$, it is never charged again as $\hleader$.  Also observe
  that if a pair is in a component $X$ whose $\hleader$ is defined,
  subsequently, it always stays in a component in which the $\hleader$
  is defined.

  \medskip
  For any $i$, we claim that the total charge to pair $(s_i, t_i)$ is
  at most $8 \ell_i$, which completes the proof. Consider any such
  pair: any charges to the pair occur when it is either $\leader$ or
  $\hleader$ of its current component. First, consider charges to
  $(s_i, t_i)$ as $\leader$ of a component. Such a charge can only
  occur when connecting some $s_j$ (or $t_j$) to $X$. Furthermore, if
  $h = \floor{\log d_j(s_j, X)} \le \class(X) = \class(i)$, the
  $\hleader(X)$ must be \emph{currently undefined}, for otherwise the
  $\hleader(X)$ would have been charged.  Subsequently, the $\hleader$
  of the component containing $(s_i,t_i)$ is always defined, and so
  $(s_i, t_i)$ will never again be charged as a $\leader(X)$ by a path
  of length in $[2^h, 2^{h+1})$. Therefore, the total charge to $(s_i,
  t_i)$ as $\leader$ of a component is at most $\sum_{h =
    1}^{\class(i)} 2^{h+1} < 2^{\class(i) + 2} \le 4 \ell_i$.

  Finally, consider charges to $(s_i, t_i)$ as $\hleader$ of a
  component. As observed above, $h \le \class(i)$.  Also for a fixed
  $h$, a pair is charged at most once as $\hleader$. Since the total
  cost charged to $(s_i, t_i)$ as $\hleader$ is at most $2^{h+1}$;
  summing over all $h \le \class(i)$, the total charge is less than
  $2^{\class(i) + 2} = 4 \ell_i$.

  Thus, the total charge to $(s_i, t_i)$ is at most $4 \ell_i + 4
  \ell_i = 8 \ell_i$, completing the proof.
  \end{proof}

\begin{lemma}\label{lem:lowerBound}
  If $\opt$ denotes the cost of an optimal solution to the instance of
  Cap-SNDP with multiple copies, then $\sum_{i=1}^k
    \ell_i \le 64 (\ceil{\log k} + 1) \opt$.
\end{lemma}
\begin{proof}
  Let $C_h$ denote $\sum_{i: \class(i) = h} \ell_i$. Clearly,
  $\sum_{i=1}^k \ell_i = \sum_h C_h$. 
  The lemma follows from the two sub-claims below:

  \noindent \textbf{Sub-Claim 1: }
  $\sum_h C_h \le (2 (\ceil{\log k}+1)) \cdot \max_h C_h$

  \noindent \textbf{Sub-Claim 2: }
  For each $h$, $C_h \le 32 \opt$.

  \begin{proofof}{Sub-Claim 1}
    Let $h' = \max_i \class(i)$. We have $C_{h'} \ge 2^{h'}$, and for
    any terminal $i$ such that $\class(i) \le h' - (\ceil{\log k}+1)$,
    we have $\ell_i \le \frac{2^{h'+1}}{2k}$. Thus, the total
    contribution from such classes is at most $\frac{2^{h'}}{k} \cdot
    k = 2^{h'}$, and hence:
    \vspace{-5mm}
    \begin{eqnarray*}
      \sum_{h = h'-\ceil{\log k}}^{h'} C_h &\ge& \frac{\sum_h C_h}{2}
      \textrm{, which implies}\\
      \max_{h' - \ceil{\log k} \le h \le h'} C_h &\ge&
      \frac{\sum_h C_h}{2 (\ceil{\log k} + 1)}.\qedhere 
    \end{eqnarray*}
  \end{proofof}

  \bigskip \def\ball{{\tt ball}}

  It remains to show Sub-Claim 2, that for each $h$, $C_h \le 32
  \opt$.  Fix $h$.  Let $\S_h$ denote the set of pairs $s_i,t_i$ such
  that $\class(i) = h$.  Our proof will go via the natural primal and
  dual relaxations for the Cap-SNDP problem. In particular, we will
  exhibit a solution to the dual relaxation of cost $C_h/32$. To do so
  we will require the following claim.  Define $\ball(s_i, r)$, a {\em
    ball} of radius $r$ around $s_i$ as containing the set of vertices
  $v$ such that $d_i(s_i,v) \le r$ and the set of edges $e = uv$ such
  that $d_i(s_i,\{u,v\}) + c_i(e) \le r$. An edge $e$ is
  \emph{partially} within the ball if $d_i(s_i, \{u,v\}) < r <
  d_i(s_i, \{u,v\}) + c_i(e)$. Subsequently, we assume for ease of
  exposition that no edges are partially contained within the balls we
  consider; this can be achieved by subdividing the edges as
  necessary. Similarly, we define $\ball(t_i, r)$, the ball of radius
  $r$ around $t_i$. Two balls are said to be \emph{disjoint} if they
  contain no common vertices.
  
  \begin{claim}\label{claim:balls}
    There exists a subset of pairs, $\S'_h\subseteq \S_h$, $|\S'_h| \ge
    |\S_h|/2$, and a collection of $|\S'_h|$ disjoint balls of radius
    $2^h/4$ centred around either $s_i$ or $t_i$, for every pair
    $(s_i,t_i)\in \S'_h$.
  \end{claim}
  
  We prove this claim later; we now use it to complete the proof of
  Sub-Claim 2.  First we describe the LP.  Let the variable $x_e$
  denote whether or not edge $e$ is in the Cap-SNDP solution. Let
  $\P_i$ be the set of paths from $s_i$ to $t_i$. For each $P \in
  \P_i$, variable $f_P$ denotes how much flow $t$ sends to the root
  along path $P$. We use $u_i(e)$ to refer to $\min\{R_i, u(e)\}$, the
  effective capacity of edge $e$ for pair $(s_i, t_i)$.

  \begin{tabular}{c|l}
    \begin{minipage}[htb]{0.45\linewidth}
      \begin{eqnarray*}
        \mbox{\bf Primal} && \hspace{-0.3in}\min \sum_{e \in E} c_e x_e  \\
        \sum_{P \in \P_i}f_{P} & \ge & R_i \quad \quad ~~~ \left(
          \forall i \in [k] \right) \\
        \sum_{P \in \P_t| e \in P} f_{P} & \le & u_i(e) x_e \quad  \left(
          \forall i \in [k], e \in E \right)\\
        x_e, f_{P} & \ge & 0\\
      \end{eqnarray*}
    \end{minipage}
    &
    \begin{minipage}[htb]{0.45\linewidth}
      \begin{eqnarray*}
        \mbox{\bf Dual} && \hspace{-0.3in}\max \sum_{t \in T} R_i \alpha_i \\
        \sum_i u_i(e) \beta_{i,e} &\le &c_e \quad \quad \quad
        ~ \left( \forall e  \in E \right) \\
        \alpha_i &\le & \sum_{e \in P} \beta_{i,e} \quad  
        \left(\forall i \in [k], P \in \P_i \right)\\
        \alpha_i, \beta_{i,e} &\ge & 0\\
      \end{eqnarray*}
    \end{minipage}
  \end{tabular}

  We now describe a feasible dual solution of value at least $C_h/32$
  using Claim \ref{claim:balls}. 
  For $(s_i,t_i) \in \S'_h$, if there is a ball $B$ around $s_i$ (or
  equivalently $t_i$), we define $\beta_{i,e} = c(e)/u_i(e)$ for each
  edge in the ball. Since the balls are disjoint, the first inequality
  of the dual is clearly satisfied.  Set $\alpha_i = 2^h/8R_i$. For
  any path $P\in \P_i$, we have
  \[ \sum_{e \in P} \beta_{i,e} 
  =  \frac{1}{R_i} \sum_{e \in P \cap B} \frac{R_i c(e)}{u_i(e)} 
  \ge \frac{1}{2R_i} \sum_{e \in P \cap B} \frac{R_i c(e)}{u(e)} + c(e)
  \ge \frac{1}{2R_i} \sum_{e \in P \cap B} c_i(e)
  \ge \frac{1}{2R_i} \frac{2^h}{4} = \alpha_i \]

  \noindent
  where the first inequality used $u_i(e) \le R_i$, the second follows
  from the definition of $c_i(e)$, and the last inequality follows
  from the definition of $\ball(s_i,2^h/4)$. Thus, $\alpha_i
  =2^h/8R_i$ is feasible along with these $\beta_{i,e}$'s. This gives
  a total dual value of

  $$\frac{2^h}{8}\cdot |\S'_h| \ge \frac{2^h}{16}\cdot |\S_h| \ge \frac{1}{32} \sum_{i\in \S_h} \ell_i = \frac{C_h}{32}$$
  where the last inequality follows from the fact that $\class(i) =
  h$. This proves the lemma modulo Claim \ref{claim:balls}, which we
  now prove.

  \begin{proofof}{Claim \ref{claim:balls}}
    We process the pairs in $\S_h$ in the order they are processed by
    the original algorithm and grow the balls.  We abuse notation and
    suppose these pairs are $(s_1,t_1),\ldots,(s_p,t_p)$. We maintain
    a collection of disjoint balls of radius $r=2^h/4$, initially
    empty.

    At stage $i$, we try to grow a ball of radius $r$ around either
    $s_i$ or $t_i$. If this is not possible, the ball around $s_i$
    intersects that around some previous terminal in $\S'_h$, say
    $s_j$; similarly, the ball around $t_i$ intersects that of a
    previous terminal, say $t_\ell$.  Let $v$ be a vertex in
    $\ball(s_i, r)$ and $\ball(s_j, r)$. We have $d_i(s_i, s_j) \le
    d_i(s_i,v) + d_i(v, s_j) \le d_i(s_i, v) + d_j(v, s_j) <
    2^h/2$. (The second inequality follows because for any $j < i$ and
    any edge $e$, $c_i(e) \le c_j(e)$.) Similarly, we have $d_i(t_i,
    t_\ell) < 2^h/2$.

    Now, we observe that $s_j$ and $t_\ell$ could not have been in the
    same component of $F$ at the beginning of iteration $i$ of {\sc
      Cap-SNDP-MC}; otherwise $d_i(s_i,t_i) \le d_i(s_i,s_j) +
    d_i(t_i,t_\ell) <2^h$, contradicting that $\class(i) =h$. But
    since $d_i(s_i, s_j) \le 2^h/2$ and $\class(i) = \class(j) = h$, we
    connect $s_i$ to the component of $s_j$ during iteration $i$;
    likewise, we connect $t_i$ to the component of $t_\ell$ during
    this iteration. Hence, at the end of the iteration, $s_i, t_i,
    s_j, t_\ell$ are all in the same component. As a result, the
    number of components of $F$ containing pairs of $\S_h$
    \emph{decreases} by at least one during the iteration.

    It is now easy to complete the proof: During any iteration of $F$
    corresponding to a pair $(s_i, t_i) \in \S_h$, the number of
    components of $F$ containing pairs of $\S_h$ can go up by at most
    one. Say that an iteration \emph{succeeds} if we can grow a ball
    of radius $r$ around either $s_i$ or $t_i$, and \emph{fails}
    otherwise.  During any iteration that fails, the number of
    components decreases by at least one; as the number of components
    is always non-negative, the number of iterations which fail is no
    more than the number which succeed. That is, $|\S'_h| \ge |\S_h -
    \S'_h|$.
  \end{proofof}
\end{proof}

Theorem~\ref{thm:multipleCopies} is now a straightforward consequence
of Lemmas~\ref{lem:ellPays} and \ref{lem:lowerBound}:

\begin{proofof}{Theorem~\ref{thm:multipleCopies}}
  The total cost of edges bought by the algorithm is at most
  $\sum_{i=1}^k \sum_{e \in F_i} c_i(e) \le 9 \sum_{i=1}^k \ell_i$, by
  Lemma~\ref{lem:ellPays}. But $\sum_{i=1}^k \ell_i \le 64(\ceil{\log
    k} + 1) \opt$, by Lemma~\ref{lem:lowerBound}, and hence the total
  cost paid by {\sc Cap-SNDP-MC} is at most $O(\log k) \opt$. 
\end{proofof}

\section{Conclusions} \label{sec:conclusions} 

In this paper we made progress on addressing the approximability of
Cap-SNDP. We gave an $O(\log n)$ approximation for the \capR problem,
which is a capacitated generalization of the well-studied min-cost
$\lambda$-edge-connected subgraph problem. Can we improve this to obtain an
$O(1)$ approximation or prove super-constant factor hardness of
approximation? We also highlighted the difficulty of Cap-SNDP by
focusing on the single pair problem, and showing both super-constant
hardness and an $\Omega(n)$ integrality gap example, even for the LP
with KC inequalities. We believe that understanding the single pair
problem is the key to understanding the general case. In particular,
we do not have a non-trivial algorithm even for instances in which the
edge capacities are either $1$ or $U$; this appears to capture much of
the difficulty of the general problem. As we noted, allowing multiple
copies of edges makes the problem easier; in practice, however, it may
be desirable to not allow too many copies of an edge to be used. It is
therefore of interest to examine the approximability of Cap-SNDP if we
allow only a small number of copies of an edge. Does the problem admit
a non-trivial approximation if we allow $O(1)$ copies or, say, $O(\log
n)$ copies? This investigation may further serve to delineate the easy
versus difficult cases of Cap-SNDP.

\paragraph{Acknowledgements:} CC's interest in capacitated network
design was inspired by questions from Matthew Andrews. He thanks
Mathew Andrews and Lisa Zhang for several useful discussions
on their work on capacitated network design for multi-commodity
flows.

\appendix

\section{Hardness of Approximation for  Cap-SNDP in Undirected Graphs}
\label{app:copiesHardness}

In this section, we prove Theorem~\ref{thm:copiesHardness} via a
reduction from the \emph{Priority Steiner Tree} problem. In the
Priority Steiner Tree problem, the input is an undirected graph
$G(V,E)$ with a cost $c(e)$ and a priority $P(e) \in \{1,2,\dots,k\}$
for each edge $e$. (We assume $k$ is the highest and $1$ the lowest
priority.) We are also given a root $r$ and a set of terminals $T
\subseteq V - \{r\}$; each terminal $t \in T$ has a desired priority
$P(t)$. The goal is to find a minimum-cost Steiner Tree in which the
unique path from each terminal $t$ to the root consists only of edges
of priority $P(t)$ or higher.\footnote{It is easy to see that a
minimum-cost subgraph containing such a path for each terminal is a
tree; given any cycle, one can remove the edge of lowest priority.}

Chuzhoy \etal \cite{CGNS} showed that one cannot approximate the
Priority Steiner Tree problem within a factor better than $\Omega(\log
\log n)$ unless $NP \subseteq DTIME(n^{\log \log \log n})$, even when
all edge costs are $0$ or $1$. Here, we show an
approximation-preserving reduction from this problem to Cap-SNDP with
multiple copies; this also applies to the basic Cap-SNDP problem, as
the copies of edges do not play a significant role in the
reduction. 

Given an instance $\script{I}_{pst}$ of Priority Steiner Tree on graph
$G(V,E)$ with edge costs in $\{0,1\}$, we construct an instance
$\script{I}_{cap}$ of Cap-SNDP defined on the graph $G$ as the
underlying graph.  Fix $R$ to be any integer greater than $2m^3$ where
$m$ is the number of edges in the graph $G$.  We now assign a capacity
of $u(e) = R^i$ to each edge $e$ with priority $P(e) = i$ in
$\script{I}_{pst}$. Each edge $e$ of cost $0$ in $\script{I}_{pst}$
has cost $c(e) = 1$ in $\script{I}_{cap}$, and each edge $e$ of cost
$1$ in $\script{I}_{pst}$ has cost $c(e) = m^2$ in
$\script{I}_{cap}$. Finally, for each terminal $t$, set $R_{tr} = R^i$
if $P(t) = i$; for every other pair of vertices $(p,q)$, $R_{pq} = 0$.

Let $C$ denotes the cost of an optimal solution to $\script{I}_{pst}$;
note that $C \le m$; we now argue that $\script{I}_{pst}$ has an
optimal solution of cost $C$ iff $\script{I}_{cap}$ has an optimal
solution of of cost between $C m^2$ and $Cm^2 + m < (C+1) m^2$. Given
a solution $E^*$ to $\script{I}_{pst}$ of cost $C$, simply select the
same edges for $\script{I}_{cap}$; the cost in $\script{I}_{cap}$ is
at most $C m^2 + m$ since in $\script{I}_{cap}$, we pay $1$ for each
edge in $E^*$ that has cost $0$ in $\script{I}_{pst}$.  This is
clearly a feasible solution to $\script{I}_{cap}$ as each terminal $t$
has a path to $r$ in $E^*$ containing only edges with priority at
least $P(t)$, which is equivalent to having capacity at least
$R_{tr}$. Conversely, given a solution $E'$ to $\script{I}_{cap}$ with
cost in $[Cm^2, (C+1)m^2)$, select a single copy of each edge in $E'$
as a solution to $\script{I}_{pst}$; clearly the total cost is at most
$C$. To see that this is a feasible solution, suppose that $E'$ did
not contain a path from some terminal $t$ to the root $r$ using edges
of priority $P(t)$ or more. Then there must be a cut separating $t$
from $r$ in which all edges of $E'$ have capacity at most $R^{P(t) -
  1}$. But since $E'$ supports a flow of $R^{P(t)}$ from $t$ to $r$,
it must use at least $R$ edges (counting with multiplicity); this
implies that the cost of $E'$ is at least $R \ge (C+1)m^2$, a
contradiction.

\bigskip
We remark that a similar reduction also proves $\Omega(\log \log n)$
hardness for the single-pair Cap-SNDP problem without multiple copies:
One can effectively encode an instance of the single-source
Fixed-Charge Network Flow (FCNF, \cite{CGNS}), very similar to
single-source Cap-SNDP with multiple copies, as an instance of
single-pair Cap-SNDP \emph{without} multiple copies: Create a new sink
$t^*$, and connect $t^*$ to each original terminal $t$ with a single
edge of cost $0$ and capacity $R_{tr}$. The only way to send flow
$\sum_{t \in T}R_{tr}$ flow from $t^*$ to the source $s$ is for each
terminal $t$ to send $R_{tr}$ to $s$. Thus,
Theorem~\ref{thm:singlePairHardness} is a simple consequence of the
$\Omega(\log \log n)$ hardness for single-source FCNF \cite{CGNS}.

\section{Omitted Proofs}
\label{app:omittedProofs}

\subsection{Proof of Theorem \ref{thm:uniform}: Near Uniform Cap-SNDP}
\label{app:nearlyUniform}

The algorithm described in Section \ref{sec:uniformReq} can be
extended to the case where requirements are \emph{nearly} uniform,
that is, if $R_{pq} \in [R,\gamma R]$ for all pairs $(p,q) \in V
\times V$. We obtain an $O(\gamma \log n)$-approximation, while
increasing the running time by a factor of $O(n^{4\gamma})$. We work
with a similar LP relaxation; for each set $S \subseteq 2^V$, we use
$R(S) = \max_{p \in S, q \not \in S}\{R_{pq}\}$ to denote the
requirement of $S$. Now, the original constraints are of the form
$$\sum_{e \in  \delta(S)} u(e) x_e \ge R(S)$$
for each set $S$, and we define the residual requirement for a set as
$R(S,A) = \min\{0,R(S)-u(A \cap \delta(S))\}$. The KC inequalities use
this new definition of $R(S,A)$.

Given a fractional solution $x$ to the KC LP, we modify the
definitions of highly fractional and nearly integral edges: An edge
$e$ is said to be nearly integral if $x_e \ge \frac{1}{40 \gamma \log
  n}$, and highly fractional otherwise. Again, for a fractional
solution $x$, we let $A_x$ denote the set of nearly integral edges;
the set $\S$ of small cuts is now $\{S \subseteq V \colon
\u(\delta(S)) \le 2 \gamma R\}$. From the cut-counting theorem, $|\S|
\le n^{4 \gamma}$.  We use $\L$ to denote the set of \emph{large}
cuts, the sets $\{S \subseteq V \colon \u(\delta(S)) > 2 \gamma R\}$.

As before, a fractional solution $x$ is \emph{good} if the original
constraints are satisfied, and the KC Inequalities are satisfied for
the set of edges $A_x$ and the sets in $\S$. These constraints can be
checked in time $O(n^{4 \gamma + 2} \log^2 n)$, so following the proof
of Lemma~\ref{lem:ell}, for constant $\gamma$, we can find a good
fractional solution in polynomial time.

The rounding and analysis proceed precisely as before: For each highly
fractional edge $e$, we select it for the final solution with
probability $40 \gamma \log n \cdot x_e$. The expected cost of this
solution is at most $O(\gamma \log n)$ times that of the optimal
integral solution, and analogously to the proofs of
Lemmas~\ref{lem:largeCuts} and \ref{lem:smallCuts}, one can show that
the solution satisfies all cuts with high probability. This completes
the proof of Theorem \ref{thm:uniform}.

\subsection{Proof of Theorem~\ref{thm:kWay}}

To prove Theorem~\ref{thm:kWay}, we work with the generalization of
\KCLP given below. For any $i$-way cut $\C$ and for any set of edges
$A$, we use $R(\C,A)$ to be $\max\{0, R_i - u(A \cap \delta(C)
\}$.\footnote{For ease of notation, we assume that for any edge $e$,
  $u(e) \le R_1$. This is not without loss of generality, but the
  proof can be trivially generalized: In the constraint for each
  $i+1$-way cut $\C$ such that $e \in \delta(\C)$, simply use the
  minimum of $u(e)$ and $R_i$.}

\vspace{-5mm}
\begin{align}
  \min ~~~ \sum_{e\in E} c(e)x_e \tag{\textrm{$k$-way} KC LP} \\ 
  \vspace{-2mm}
  \forall i, \forall \textrm{$i$-way cuts } \C, ~~~~~~~~~~~~ \sum_{e\in \delta(\C)} u(e)x_e \ge R_i \tag{Original Constraints} \\
  \vspace{-3mm} \forall A\subseteq E, \forall i, \forall \textrm{$i$-way cuts } \C, ~~~~~~~~~
  \sum_{e\in \delta(\C)\setminus A} \min\{u(e),R(\C,A)\}x_e \ge R(\C,A)
  \tag{KC-inequalities} \\
  \vspace{-3mm} \forall e\in E,~~~~~~~~~~~~~~~~~~~~~ 0 \le x_e \le 1
  \notag
\end{align}

As before, given a fractional solution $x$ to this LP, we define $A_x$
(the set of nearly integral edges) to be $\{e \in E \colon x_e \ge
\frac{1}{40 k \log n}\}$. Define $\hat{u}(e) = u(e)x_e$ to be the
fractional capacity on the edges.  Let $\S_i := \{\C: \C \textrm{ is
  an $i+1$-way cut and } \u(\delta(C)) \le 2R_i\}$.  The solution $x$ is
said to be {\em good} if it satisfies the following three conditions:

\begin{enumerate}
\item[(a)] If the capacity of $e$ is $\u(e)$, the capacity of any
  $i+1$-way cut in $G$ is at least $R_i$; equivalently $x$ satisfies the
  original constraints.
\item[(b)] The KC inequalities are satisfied for the set $A_x$ and the
  sets in $\S_i$, for each $1 \le i \le k-1$. Note that if (a) is
  satisfied, then by Lemma~\ref{lem:kWayCutCounting}, $|\S_i| \le n^{4i}$.
\item[(c)] $\sum_{e\in E} c(e)x_e$ is at most the value of the optimum
  solution to the linear program ($k$-way KC LP).
\end{enumerate}

Following the proof of Lemma~\ref{lem:ell}, it is straightforward to
verify that there is a randomized algorithm that computes a good
fractional solution with high probability in $n^{O(k)}$ time.

Once we have a good fractional solution, our algorithm is to select
$A_x$, the set of nearly integral edges, and to select each highly
fractional edge $e \in E \setminus A_x$ with probability $40 k \log n
\cdot x_e$. If $F^*$ denotes the highly fractional edges that were
selected, we return the solution $A_x \cup F^*$. As before, it is
trivial to see that the expected cost of this solution is $O(k \log
n)$ times that of the optimal integral solution. 

We show below that for any $i \le k-1$, we satisfy all $i+1$-way cuts
with high probability; taking the union bound over the $k-1$ choices
of $i$ yields the theorem. 

As in Lemmas~\ref{lem:largeCuts} and \ref{lem:smallCuts}, we
separately consider the ``large'' and ``small'' $i+1$-way cuts.
First, consider any small cut $\C$ in $\S_i$. From the Chernoff bound
(Lemma~\ref{lem:Chernoff}) and the KC inequality for $\C$ and $A_x$,
it follows that the probability we fail to satisfy $\C$ is at most
$1/n^{19k}$. From the cut-counting Lemma~\ref{lem:kWayCutCounting},
there are at most $n^{4i} < n^{4k}$ such small cuts, so we satisfy all
the small $i+1$ way cuts with probability at least
$1-\frac{1}{n^{15k}}$.

For the large $i+1$-way cuts $\L$, we separately consider cuts of
differing capacities. For each $j \ge 2$, let $\L(j)$ denote the
$i+1$-way cuts $\C$ such that $jR_i \le \u(\C) \le (j+1)
R_i$. Consider any cut $\C \in \L_j$; if $u(A_x \cap \delta(C)) \ge
R_i$, then the cut $\C$ is clearly satisfied. Otherwise,
$\u(\delta(\C) \setminus A_x) \ge (j-1)R_i$. But since we selected
each edge $e$ in $\delta(\C) \setminus A_x$ for $F^*$ with probability
$40 k \log n \cdot x_e$, the Chernoff bound implies that we do not
satisfy $\C$ with probability at most $\frac{1}{n^{19 k (j-1)}}$.  The
cut-counting Lemma~\ref{lem:kWayCutCounting} implies there are most
$n^{2i(j+1)} < n^{2k(j+1)}$ such cuts, so we fail to satisfy any cut
in $\L(j)$ with probability at most $n^{21-17j}$. Taking the union
bound over all $j$, the failure probability is at most $2n^{-13}$.

\end{document}